\documentclass[conference,letterpaper]{IEEEtran}

\usepackage[utf8]{inputenc} 
\usepackage[T1]{fontenc}
\usepackage{url}
\usepackage{ifthen}
\usepackage{cite}
\usepackage[cmex10]{amsmath} 

\usepackage{amssymb,xcolor,hyperref}
\usepackage{amsthm}%
\usepackage{mathrsfs}%
\usepackage{bm,hyperref} %
\usepackage[ruled, lined, linesnumbered, commentsnumbered, longend]{algorithm2e}

\newtheorem{thm}{Theorem} 
\newtheorem{lem}{Lemma}
\newtheorem{cor}{Corollary}
\newtheorem{prop}{Proposition}

\newtheorem{defn}{Definition}

\newcommand{\NS}{\hat{N}}

\begin{document}
\title{Ratio Attack on G+G Convoluted Gaussian Signature} 


\author{%
  \IEEEauthorblockN{Chik How Tan,  Theo Fanuela Prabowo and Wei Guo Foo}
  \IEEEauthorblockA{Temasek Laboratories, National University of Singapore\\
                    5A Engineering Drive 1, \#09-02, Singapore 117411\\
                    Email: \{tsltch, tsltfp, fwg\}@nus.edu.sg} }

\maketitle


\begin{abstract}
   A lattice-based signature, called G+G convoluted Gaussian signature \cite{DPS23} was proposed in ASIACRYPT 2023 and was proved secure in the quantum random oracle model. In this paper, we propose a ratio attack on the G+G convoluted Gaussian signature \cite{DPS23} to recover the secret key and comment on the revised eprint paper \cite{DPS23e}.  The attack exploits the fact, proved in this paper, that the secret key can be obtained from the expected value of the ratio of signatures which follows a truncated Cauchy distribution. Moreover, we also compute the number of signatures required to successfully recover the secret key. Furthermore, we simulate the ratio attack in Sagemath with a few different parameters as a proof-of-concept of the ratio attack. In addition, although the revised signature in the revised eprint paper \cite{DPS23e} is secure against the ratio attack, we found that either a valid signature cannot be produced or a signature can be forged easily for their given parameters in \cite{DPS23e}.
\end{abstract}

\section{Introduction}

Recently, NIST has standardized a lattice-based signature called ML-DSA \cite{ML-DSA24} which was the Dilithium signature submitted to NIST PQC competition. Dilithium signature is based on Fiat-Shamir transformation \cite{FS86} and adapted the technique from Lyubashevsky's signature scheme \cite{Lyu09,Lyu12} with rejection sampling or abort. 
The rejection sampling or abort means that the signature will be rejected if some conditions are not satisfied and the generation of a new signature is repeated until the conditions are satisfied.  A number of lattice-based signatures are constructed using this approach, such as Dilithium \cite{DKL18}, qTESLA \cite{ABB20}, etc. 

Since the original lattice-based signature based on Fiat-Shamir transformation is subjected to statistical attack, the rejection sampling or abort technique was proposed by Lyubashevsky \cite{Lyu09,Lyu12} to defend against the statistical attack.  The security of such Dilithium-like signature schemes has been analyzed in \cite{PT24}. The paper \cite{PT24} determines which parameters are subjected to the attack and which ones are not.  Recently,  Devevey et al. \cite{DPS23} proposed a lattice-based signature called the G+G convoluted Gaussian signature without rejection sampling or abort.  The G+G convoluted Gaussian signature \cite{DPS23} is secure against the known statistical attack as 
the ephemeral key is sampled dependent on the known hash value $\bm c$ in a way that hides information in the signature.
In this paper, we propose a new attack on the G+G convoluted Gaussian signature, called ratio attack.  The ratio attack explores the correlation among the signatures and takes the average of the ratio of two signatures. Furthermore, the signatures follow a normal distribution and the ratio of two normal distributions is a Cauchy distribution. However, Cauchy distribution does not have finite expected value and standard deviation.  By exploring the truncated Cauchy distribution for which the Central Limit Theorem does apply, we can find a formula relating the average of the ratio of two signatures and the secret key.  Therefore, we can successfully recover the secret key using the ratio attack.


The rest of the paper is organized as follows. In Section \ref{SectPrelim}, we setup some notations and provide some statistical results relevant to our attack, for example, multivariate Gaussian distribution, Cauchy distribution and truncated Cauchy distribution, etc.  In Section \ref{G+G}, we review the G+G convoluted Gaussian signature \cite{DPS23}. In Section \ref{SecRatioAttack}, we introduce the ratio attack and prove a formula relating the expected value of the ratio of two signatures and the secret key.  We also provide an approximation of the required number of signatures to successfully recover the secret key.  In section \ref{SectAttack}, we give more details on the ratio attack on the G+G convoluted Gaussian signature and implement the ratio attack in Sagemath   on some scaled-down parameters as a proof-of-concept. 
In section \ref{Sect_com_DPS23e}, we examine and comment on the revised signature in the revised eprint paper \cite{DPS23e}.  Finally, we give some concluding remarks in Section \ref{SectConclusion}.

\section{Preliminaries}\label{SectPrelim}

\subsection{Notations}

Let $q$ be an odd prime and $n$ be a power of two.
Let $\mathbb{Z}_q = \mathbb{Z}/q\mathbb{Z}$ denote the quotient ring of integers modulo $q$, which we represent as
$\mathbb{Z}_q= \{-\frac{q-1}{2}, \ldots,  \frac{q-1}{2}\}$.
Let $\mathcal{R}$, $\mathcal{R}_q$, $\mathcal{R}_{2q}$ denote the rings $\mathbb{Z}[x]/(x^n + 1)$, $\mathbb{Z}_q[x]/(x^n + 1)$, and $\mathbb{Z}_{2q}[x]/(x^n+1)$ respectively.
For $\bm a = \sum_{i=0}^{n-1} a_i x^i \in \mathcal{R}$, we define $\| \bm a \|_{\infty}:= \max \{|a_i| \mid 0 \le i \le n-1\}$ and $\|\bm a \|:= \sqrt{\sum_{i=0}^{n-1}a_i^2}$.


\begin{defn}[Skew Circulant Matrix] \label{cirmat}
   For ${\bf v}=(v_0, \ldots, v_{n-1}) \in \mathbb{Z}^n$, the skew circulant matrix defined by $\mathbf{v}$ is  
\[ V:= \begin{bmatrix} v_0 & -v_{n-1} & \ldots & -v_1 \\ v_1 & v_0 & \ldots & -v_2 \\ \vdots & \vdots & \ddots & \vdots \\ v_{n-1} & v_{n-2} & \ldots & v_0 \end{bmatrix} \in \mathbb{Z}^{n \times n}. \]
\end{defn}


For $\bm u = \sum_{i=0}^{n-1} u_i x^i,\bm v = \sum_{j=0}^{n-1} v_j x^j \in \mathcal{R}$, the product $\sum_{l=0}^{n-1}w_lx^l = \bm{ w}=\bm{uv}$ can be computed as $\mathbf{ w} = {\bf u}V^T$, where $\mathbf{u}=(u_0, \ldots, u_{n-1}), \mathbf{v}=(v_0, \ldots, v_{n-1})$, and $ \mathbf{w}=(w_0, \ldots, w_{n-1})$.
For $l=0, \ldots, n-1$, we have
$$w_l = \sum_{i+j=l \bmod n} \epsilon_{i,j}   u_iv_j, \quad \text{where} \ \epsilon_{i,j} := \begin{cases} 
                      1 & \mbox{if }i+j < n, \\
                      -1 & \mbox{if }i+j \ge n.
                    \end{cases}$$

\subsection{Some Statistical Results}

Let $U$ be a random variable. We denote its expected value and variance by $\mathbb{E}(U)$ and $\mathbb{V}(U)$ respectively. 


\begin{lem}[{\cite{TP22}}]  \label{sum_prod_dist}
  Let $U, V$ be random variables with mean $\mathbb{E}(U), \mathbb{E}(V)$ and variance $\mathbb{V}(U), \mathbb{V}(V)$. Then, $\mathbb{E}(U \pm V)=\mathbb{E}(U) \pm \mathbb{E}(V)$. If moreover $U$ and $V$ are independent, then
 \begin{itemize}
     \item[{\rm (a)}] $\mathbb{V}(U \pm V)=\mathbb{V}(U) + \mathbb{V}(V)$,
     \item[{\rm (b)}] $\mathbb{E}(UV)=\mathbb{E}(U)\mathbb{E}(V)$ and $\mathbb{V}(UV)=(\mathbb{V}(U) + \mathbb{E}(U)^2) \times (\mathbb{V}(V) + \mathbb{E}(V)^2)-\mathbb{E}(U)^2\mathbb{E}(V)^2$.
 \end{itemize}
\end{lem}



\begin{lem}[{\cite{TP22}}] \label{uniform_dist}
   Let $b$ be a positive integer. Then, the mean and the variance of $\mathcal{U}([-b,b]\cap \mathbb{Z})$ (i.e. the uniform distribution on $[-b,b] \cap \mathbb{Z}$) are $0$ and $\frac{b(b+1)}{3}$ respectively. 
\end{lem}


The (univariate) normal/Gaussian distribution $\mathcal{N}(\mu,\sigma^2)$ with mean $\mu$ and standard deviation $\sigma$ has probability density function given by $\rho_{\sigma}(t):=\frac{1}{\sqrt{2\pi \sigma^2}} e^{-\frac{(t-\mu)^2}{2\sigma^2}}$ for $t \in \mathbb{R}$. If $\mu=0$ and $\sigma=1$, the resulting distribution $\mathcal{N}(0,1)$ is called the standard normal distribution.


\begin{thm}[{\cite[Theorem 2.23]{Pan16}} (Central Limit Theorem)] \label{CLT}
  Let $U_1, U_2,$ $ \ldots,U_{\bar{n}}$ be  independent and identically distributed random variables 
  with mean $\mu$ and standard deviation $\sigma$.
  Let $\overline{U}:=\frac{1}{\bar{n}}\sum_{i=1}^{\bar{n}} U_i$. Then $\sqrt{\bar{n}}(\overline{U}-\mu) $ approximates to the normal distribution $\mathcal{N}(0,\sigma^2)$ with mean $0$ and standard deviation $\sigma$, that is,
\[ \lim_{\bar{n} \rightarrow \infty} \Pr\left(\frac{\overline{U}-\mu}{\sigma / \sqrt{\bar{n}}} \le \omega \right) = \Phi(\omega),  \] 
 where $\Phi(\omega):=\displaystyle\frac{1}{\sqrt{2\pi}} \int_{-\infty}^{ \omega} e^{-t^2/2} {\rm d}t.$   
\end{thm}

\begin{thm}[{\cite[Theorem 2]{PT24}}] \label{ThmNSample}
  Let $U_1,U_2,\ldots, U_{\NS}$ be independent and identically distributed random variables with mean $\mu$ and standard deviation $\sigma$. Let $\overline{U}:=\frac{1}{\NS} \sum_{i=1}^{\NS} U_i$ and $d>0$. Then, 
  the required number $\NS$ of samples such that $|\overline{U}-\mu|\le d$ with probability $\Phi(\omega)-\Phi(-\omega)$ is $\NS=  \left(\frac{\omega \sigma}{d}\right)^2$,
  where $\Phi(\omega):=\displaystyle\frac{1}{\sqrt{2\pi}} \int_{-\infty}^{ \omega} e^{-t^2/2} {\rm d}t$.
\end{thm}




We list some values of $\omega$ with the corresponding probability $(\Phi(\omega)-\Phi(-\omega))$ in the following Table \ref{ZProb}.

\begin{table}[h] 
\caption{Some values of $\omega$ with their corresponding probabilities $\Phi(\omega)- \Phi(-\omega)$}  \label{ZProb}
\centering
{ \scriptsize
\begin{tabular}{|c||c|c|c|c|c|c|c|c|} \hline
 $\omega$  & 1.96 & 2.326 & 2.576  & 2.807 & 3.090 & 3.2905 & 3.8905 & 4.4171 \\ \hline
  Prob. & 0.95 & 0.98  & 0.99   & 0.995 & 0.998 & 0.999 & 0.9999 & 0.99999\\ \hline
\end{tabular}
}
\end{table}


\begin{lem}[{\cite{PT24}}] \label{prod-uvw}
Let $\bm{ u}, \bm{v} \in \mathcal{R}$ and suppose each coordinates $u_i, v_i$ of ${\bm u}$ and $\bm v$ are independently distributed random variables  with  mean $\mu_u=\mu_v=0$ and variance $\sigma_{u}^2, \sigma_{v}^2$ respectively.  Then, each coordinate of $\bm{uv}$ approximates to $\mathcal{N}(0,n  \sigma_{u}^2\sigma_{v}^2)$. 
\end{lem}


\begin{lem}[{\cite{CDS03}}]  \label{L12} 
  For $t>2$, $Z \sim {\mathcal{N}(0,\sigma^2)}$, then 
\[  {\rm Pr}[\, |z| > t \sigma \; | \; z \leftarrow Z \,] \le \frac{1}{2}(e^{-t^2} +e^{-\frac{t^2}{2}} ). \]
\end{lem}


\subsection{Multivariate Gaussian Distributions}

We define the multivariate Gaussian/normal distribution as follows.
For $\mathbf{c} = (c_0, c_1, \ldots, c_{n-1}) \in \mathbb{R}^n$ and a positive-definite symmetric matrix
$\Sigma \in \mathbb{R}^{n \times n}$, the multivariate normal distribution centered at $\mathbf{c}$ with covariance parameter $\Sigma$ is denoted by $\mathcal{N}_{\mathbb{R}^n, \Sigma, \mathbf{c}}$ and has probability density function given by $\rho_{\Sigma, \mathbf{c}}(\mathbf{t}):=\frac{1}{\sqrt{(2\pi)^n \det(\Sigma)}} \exp(-\frac12 (\mathbf{t}-\mathbf{c})^T \Sigma^{-1} (\mathbf{t}-\mathbf{c} ))$ for $\mathbf{t} \in \mathbb{R}^n$.
We remark that the univariate normal distribution $\mathcal{N}(\mu, \sigma^2)$ can be viewed as $\mathcal{N}_{\mathbb{R},\sigma^2, \mu}$.

The following results relate the univariate and multivariate normal distributions.

\begin{lem}[{{\cite[Theorems 3.3.1 and 3.3.2]{Tong90}}}]\label{LemMultivariate}   \hfill
    \begin{itemize} 
        \item[(i)] Suppose $t_i$ follows a univariate normal distribution $\mathcal{N}(\mu_i, \sigma_i^2)$ for each $0 \le i \le n-1$. Then $\mathbf{t}=(t_0, t_1, \ldots, t_{n-1}) \in \mathbb{R}^n$ follows the multivariate normal distribution $\mathcal{N}_{\mathbb{R}^n,\Sigma, \bm{\mu}}$, where $\Sigma = {\rm diag}(\sigma_0^2, \sigma_1^2, \ldots, \sigma_{n-1}^2)$ and $\bm{\mu} = (\mu_0, \mu_1, \ldots, \mu_{n-1})$.
        \item[(ii)] If $\mathbf{t}=(t_0,t_1, \ldots, t_{n-1})$ follows a multivariate normal distribution $\mathcal{N}_{\mathbb{R}^n,\Sigma,\mathbf{c}}$, then $t_i$ follows $\mathcal{N}(c_i, \sigma_i^2)$ for all $0 \le i \le n-1$, where $\mathbf{c}=(c_0,c_1, \ldots, c_{n-1})$ and $\sigma_i$ is the $(i,i)$-entry of $\Sigma$. If moreover $\Sigma = \sigma^2 I_n$ for some $\sigma \in \mathbb{R}$ (where $I_n$ is the $n \times n$ identity matrix), then $t_0, t_1, \ldots, t_{n-1}$ are independent.
    \end{itemize}
\end{lem}

The discrete (multivariate) Gaussian distribution $\mathcal{D}_{\mathbb{Z}^n,\Sigma, \mathbf{c}}$ is the distribution obtained by restricting the support of $\mathcal{N}_{\mathbb{R}^n, \Sigma, \mathbf{c}}$ to the set $\mathbb{Z}^n$.
For any $\mathbf{t} \in \mathbb{Z}^n$, the probability that $\mathbf{t}$ appears is proportionate to $\rho_{\Sigma,\mathbf{c}}(\mathbf{t})$. So, the discrete Gaussian distribution $\mathcal{D}_{\mathbb{Z}^n,\Sigma, \mathbf{c}}$ has probability mass function  given by $\rho_{\Sigma,\mathbf{c}}'(\mathbf{t}):= \displaystyle\frac{\rho_{\Sigma,\mathbf{c}}(\mathbf{t})}{\sum_{\mathbf{y} \in \mathbb{Z}^n} \rho_{\Sigma,\mathbf{c}}(\mathbf{y})}$ for $\mathbf{t} \in \mathbb{Z}^n$.

The notion of discrete multivariate Gaussian distribution can be extended to a distribution over $\mathcal{R}=\mathbb{Z}[x]/(x^n+1)$. For $\bm c = \sum_{i=0}^{n-1}c_ix^i \in \mathcal{R}$ and a positive-definite symmetric matrix $\Sigma \in \mathbb{R}^{n \times n}$, we define $$\mathcal{D}_{\mathcal{R}, \Sigma, \bm c}:= \left\{ \sum_{i=0}^{n-1} t_i x^i \in \mathcal{R} \mid (t_0, t_1, \ldots, t_{n-1}) \leftarrow \mathcal{D}_{\mathbb{Z}^n,\Sigma,\mathbf{c}} \right\},$$
where $\mathbf{c} = (c_0, c_1, \ldots, c_{n-1})$.

\subsection{Cauchy Distribution and Truncated Cauchy Distribution}

The Cauchy distribution $\mathcal{C}(\alpha,\beta)$ has density function
\[ f_X(x)=\frac{1}{\beta \pi \, (1+(\frac{x-\alpha}{\beta})^2)}, \;\ \; -\infty \le x \le \infty, \]
where $\alpha \in \mathbb{R}$ and $ \beta >0$ are the location and scale parameter respectively.

It is well known that the mean and the second moment of a Cauchy distribution do not exist. Let $X_i$, $i=1,\ldots, \bar{n}$ be independent and identically distributed random variables, each following the Cauchy distribution $\mathcal{C}(\alpha,\beta)$ and $\bar{X}=\frac{1}{\bar{n}}\sum_{i=1}^{\bar{n}} X_i$. As the Cauchy distribution does not have finite mean, the Central Limit Theorem for the asymptotic normality of $\bar{X}$ does not apply. 


The truncated Cauchy distribution \cite[Section 2]{B-I13} is obtained by restricting the Cauchy distribution to a finite interval $I=[\alpha-L,\alpha+L]$, symmetric with respect to $x=\alpha$, where $L >0$. The truncated Cauchy distribution $\mathcal{C}(\alpha,\beta \, | \,[\alpha-L,\alpha+L])$ has density function given by 
\[ f_L(x)=\frac{1}{2\arctan(\frac{L}{\beta}) \; \beta(1+(\frac{x-\alpha}{\beta})^2)}, \; \; \; \alpha-L \le x \le a+L. \]

If $L \longrightarrow \infty$, then $\arctan(\frac{L}{\beta})$ tends to $\frac{\pi}{2}$ and $\mathcal{C}(\alpha,\beta \, | \,[\alpha-L,\alpha+L])$ becomes the Cauchy distribution $\mathcal{C}(\alpha,\beta)$.

The truncated Cauchy distribution $\mathcal{C}(\alpha,\beta \, | \,[\alpha-L,\alpha+L])=:X$ has finite expected value and variance, which are 
\[ \mathbb{E}(X)=\alpha, \hspace{0.5cm} \mathbb{V}(X)=\frac{\beta L}{\arctan(\frac{L}{\beta})}-\beta^2. \] 

\begin{prop} \label{truncated_cauchy_to_normal}
Let $X:= \mathcal{C}(\alpha,\beta \, | \,[\alpha-L,\alpha+L])$ be the truncated Cauchy distribution. Let  $Y_i$, $i=1,\ldots, \bar{n}$ be independent and identically distributed random variables, where each $Y_i \thicksim X$. If $\bar{Y}=\frac{1}{\bar{n}}\sum_{i=1}^{\bar{n}} Y_i$, then the Central Limit Theorem for the asymptotic normality of $\bar{Y}$ does apply and $\sqrt{\bar{n}}(\bar{Y}-\alpha)$ approximates to the normal distribution $\mathcal{N}(0,\mathbb{V}(X))$.
\end{prop}
 

\begin{prop} {\rm \cite{RD}} \label{ratio_normal}
Let $Y,Z$ be two correlated normal distributions $Y \thicksim \mathcal{N}(\mu_Y, \sigma_Y^2)$ and $Z \thicksim \mathcal{N}(\mu_Z, \sigma_Z^2)$, where $\mu_Y=\mu_Z=0$. Then, the ratio $X=\frac{Y}{Z}$ of $Y$ and $Z$ is a Cauchy distribution $\mathcal{C}(\alpha,\beta)$, where 
\[ \alpha= \rho \frac{\sigma_Y}{\sigma_Z}, \;  \beta=\frac{\sigma_Y}{\sigma_Z} \sqrt{1-\rho^2}, \;  \rho=\frac{\mathbb{E}\left((Y-\mu_Y)(Z-\mu_Z)\right)}{\sigma_Y \sigma_Z}.  \] 
\end{prop}


\section{The G+G Convoluted Gaussian Signature} \label{G+G}

This section briefly reviews the G+G convoluted Gaussian signature proposed in \cite{DPS23}.
For $\eta >0$, denote $\chi_{\eta}:=\mathcal{U}(\{\bm s \in \mathcal{R} \mid \|\bm s\|_{\infty} \le \eta\})$, i.e. the uniform distribution on $\{\bm s \in \mathcal{R} \mid \|\bm s\|_{\infty} \le \eta\}$.  Let $\mathbf{j}:=(\bm 1,0,\ldots,0) \in \mathcal{R}^m$.
We  define $\mathcal{C}$ to be the set $\mathcal{C}:=\{ \bm a = \sum_{i=0}^{n-1} a_ix^i \in \mathcal{R} \mid  a_i \in \{0,1\} \text{ for } 0 \le i < n \}$.
For $\mathbf{s} = (\bm s_0, \bm s_1, \ldots, \bm s_{k-1}) \in \mathcal{R}^k$ and $\sigma, \sigma_u > 0$, we define 
$$\mathbf{\Sigma}(\mathbf{s},\sigma,\sigma_u):=\sigma^2 I_{nk} - \sigma_u^2 {\rm circ}(\mathbf{s}){\rm circ}(\mathbf{s})^T,$$
where $I_{nk}$ is the $nk \times nk$ identity matrix and ${\rm circ}(\mathbf{s})$ is a skew circulant matrix.

Given the public parameters consisting of some positive integers $n,q,m,k$ with $k>m+1$ and some real numbers $B_s,\sigma_u,\sigma, B_z$, the G+G convoluted Gaussian signature (G+G CGS) is described in Algorithm \ref{AlgoKeyGen}, \ref{AlgoSign} and \ref{AlgoVerify}.


\begin{algorithm}
	\SetAlgoNoEnd
    \SetAlgoNoLine
    \SetNlSty{}{}{}
    \SetNlSkip{0.5em}
    \SetKwInOut{KwIn}{Input}
    \SetKwInOut{KwOut}{Output}

    \KwIn{security parameter $\lambda$}
    \KwOut{$pk= \mathbf{A}$ and $sk = \mathbf{s}$}

    Choose $\mathbf{s}_1 \leftarrow \chi_{\eta}^{k-m-1}$ and $\mathbf{s}_2 \leftarrow \chi_{\eta}^m$ 
    
    Set $\mathbf{s}:=(\bm{1} \mid \mathbf{s}_1 \mid \mathbf{s}_2) \in \mathcal{R}_{2q}^k$
    
    {\bf if} $\| \mathbf{s}\| \ge B_s$, {\bf then} repeat from Step 1
    
    Choose $\mathbf{A}_0 \leftarrow \mathcal{U}(\mathcal{R}_q^{m \times (k-m-1)})$ 
    
    Compute $\mathbf{b}^T:= \mathbf{A}_0\mathbf{s}_1^T + \mathbf{s}_2^T \bmod{q}$ 
    
    Set $\mathbf{A}:=(-2\mathbf{b}^T + q\mathbf{j}^T \mid 2\mathbf{A}_0 \mid 2 \mathbf{I}_m)$ \tcp{\rm where $\mathbf{I}_m$ is the $m \times m$ identity matrix} 
    
    The public key is $pk= \mathbf{A}$ and the secret key is $sk = \mathbf{s}$

    \caption{Key Generation of Generic G+G CGS}\label{AlgoKeyGen}
\end{algorithm}


\begin{algorithm}
    \SetAlgoNoEnd
    \SetAlgoNoLine
    \SetNlSty{}{}{}
    \SetNlSkip{0.5em}
    \SetKwInOut{KwIn}{Input}
    \SetKwInOut{KwOut}{Output}

    \KwIn{message ${\sf m}$, $pk= \mathbf{A}$ and $sk = \mathbf{s}$}
    \KwOut{signature $\mathfrak{S}$}
    
    Choose $\mathbf{y} \leftarrow \mathcal{D}_{\mathcal{R}^k, \mathbf{\Sigma}(\mathbf{s},\sigma, \sigma_u) , \bm 0}$ 
    
    Compute $\mathbf{v}  := \mathbf{Ay} \bmod{2q}$ 
    
    Compute $\bm{c}:=\mathcal{H}(\mathbf{v}, {\sf m} ) \in \mathcal{C}$ 
    
    Choose $\bm{u}\leftarrow \mathcal{D}_{\mathcal{R},\sigma_u^2 I_n,-\bm c/2}$ \tcp{\rm where $I_n$ is the $n \times n$ identity matrix} 
    Compute $\mathbf{z}:= \mathbf{y} + (2 \bm u + \bm c) \mathbf{s}$ 
    
    The signature is $\mathfrak{S} = (\mathbf{z}, \bm{c})$

    \caption{Signing of the Generic G+G CGS}\label{AlgoSign}
\end{algorithm}


\begin{algorithm}
    \SetAlgoNoEnd
    \SetAlgoNoLine
    \SetNlSty{}{}{}
    \SetNlSkip{0.5em}
    \SetKwInOut{KwIn}{Input}
    \SetKwInOut{KwOut}{Output}

    \KwIn{message ${\sf m}$, $pk= \mathbf{A}$, signature $\mathfrak{S} = (\mathbf{z}, \bm{c})$}
    \KwOut{validity of the signature}
    
    Compute $\mathbf{v}:= \mathbf{Az}^T - q\bm c \mathbf{j}^T \bmod{2q}$ 
    
   {\bf if} $\mathcal{H}(\mathbf{v}, {\sf m}) = \bm  c$ {\rm and} $\|\mathbf{z}\| \le B_z$ {\bf then} signature is valid 
   
   {\bf else} signature is invalid

    \caption{Verification of the Generic G+G CGS}\label{AlgoVerify}
\end{algorithm}


Although a concrete example of the module-LWE variant was proposed in the paper \cite{DPS23} to realize the generic G+G convoluted Gaussian signature, its signature is $\mathbf{z} = \mathbf{y} + (\bm \zeta \bm u + \bm c) \mathbf{s}$, where $\bm \zeta=1+x^{n/2} \in \mathcal{R}$ and the mean of the coordinates of $\bm \zeta \bm u + \bm c$ are not all zero. This concrete instantiation is different from that of the above generic signature $\mathbf{z}= \mathbf{y} + (2\bm u + \bm  c) \mathbf{s}$.  Therefore, there is no concrete parameters given for the generic G+G convoluted Gaussian signature.

\section{Ratio Attacks on G+G  Convoluted Gaussian Signatures}\label{SecRatioAttack}

Recall that in the G+G convoluted Gaussian signature, we have $\mathbf{z} = \mathbf{y} + (2 \bm u + \bm c) \mathbf{s}$ in $\mathcal{R}^k$, where $\bm u, \bm c \in \mathcal{R}$, $\mathbf{y}= (\bm y_0, \ldots, \bm y_{k-1}) \in \mathcal{R}^k$; $\mathbf{s} = (\bm s_0, \bm s_1, \ldots, \bm s_{k-1}) \in \mathcal{R}^k$ is the secret key, and $\mathbf{z} = (\bm z_0, \ldots, \bm z_{k-1}) \in \mathcal{R}^k$. 
It is noted that $\bm s_0=\bm 1 \in \mathcal{R}$ and so $\bm z_0=\bm y_0+(2 \bm u + \bm c)$. 
Since $\mathbf{y}$  and $\bm u$ are sampled from the multivariate normal distributions $\mathcal{D}_{\mathcal{R}^k, \mathbf{\Sigma}(\mathbf{s},\sigma, \sigma_u) , \bm 0}$ and $\mathcal{D}_{\mathcal{R},\sigma_u^2 I_n,-\bm c/2}$ respectively, then $\bm z_i = \bm y_i + (2 \bm u + \bm c) \bm s_i$ also follows a normal distribution for all $0 \le i \le k-1$.  
Let $\bm a_i=\sum_{j=0}^{n-1} a_{i,j} x^j \in \mathcal{R}$ for  $1\le i \le k-1$ and $\bm a=\bm y$ or $\bm z$.  
Let $\bm b=\sum_{j=0}^{n-1 }b_j  x^j \in \mathcal{R}$ for $\bm b=\bm u$ or $\bm c$.  
Define $w_j:=2u_j+c_j$ for $0 \le j \le n-1$. Then $z_{0,0}=y_{0,0}+w_0$ and \begin{equation}\label{EqZ}z_{i,j}=y_{i,j}+w_0s_{i,j}+\sum_{l+m=j \bmod n, l \ne 0} \varepsilon_{l,m} w_ls_{i,m}\end{equation} for $0 \le l,m \le n-1$, $1 \le i \le k-1$ and $0 \le j \le n-1$, where $\varepsilon_{l,m}=\left \{ \begin{array}{ll} 1 & \mbox{if $l+m < n$}, \\ -1 & \mbox{if $l+m \ge n$.} \end{array} \right. $


The ratio attack takes the ratio $\frac{Y}{Z}$ of two correlated distributions $Y$ and $Z$ and finds its expected value which is related to the secret key. In our case, we consider the ratio $\frac{Z_{i,j}}{Z_{0,0}}$ of two distributions $Z_{i,j}$ and $Z_{0,0}$ for $1 \le i \le k-1$ and $0 \le j \le n-1$, where $Z_{0.0}$ and $Z_{i,j}$ are the distribution for $z_{0,0}$ and $z_{i,j}$ respectively. The goal of this ratio attack is to find an exact formula relating the expected value $\mathbb{E}(\frac{Z_{i,j}}{Z_{0,0}})$ and the secret key $s_{i,j}$ for $1 \le i \le k-1$ and $0 \le j \le n-1$. 


Recall that $\mathcal{C}:=\{ \bm a = \sum_{i=0}^{n-1} a_ix^i \in \mathcal{R} \mid  a_i \in \{0,1\} \text{ for } 0 \le i < n \}$ and $\mathcal{U}(\mathcal{C})$ is the uniform distribution on $\mathcal{C}$; and $\bm s_i \leftarrow \chi_{\eta}$ for $1 \le i \le k-1$. 
We denote $V$ to be a distribution and $v$ is an element from distribution $V$. Then, we have the following result.


\begin{lem}\label{LemEc}
    If $\bm c \leftarrow \mathcal{U}(\mathcal{C})$, then $\mathbb{E}(C_j) = \frac12$ for $0 \le j \le n-1$.
\end{lem}


\begin{lem}\label{LemEV}
    Suppose $\bm c \leftarrow \mathcal{U}(\mathcal{C})$,  $\bm u \leftarrow \mathcal{D}_{\mathcal{R},\sigma_u^2 I_n,-\bm c/2}$, $\mathbf{y} \leftarrow  \mathcal{D}_{\mathcal{R}^k, \mathbf{\Sigma}(\mathbf{s},\sigma, \sigma_u) , \bm 0}$ and $\| \mathbf{s} \|_{\infty} \le \eta$ . Then, for $0 \le j \le n-1$ and $1 \le i \le k-1$,
   
\noindent {\rm (i)} $\mathbb{E}(U_j) = -\frac12 \mathbb{E}(C_j)$,
       
\noindent {\rm (ii)} $\mathbb{E}(2 U_j + C_j) = 0$,
       
\noindent {\rm (iii)} $\mathbb{V}(2 U_j + C_j) = 4 \sigma_u^2$,
       

\noindent {\rm (iv)} $\mathbb{V}(Z_{0,j}) = \sigma^2 + 3 \sigma_u^2$, 
       

\noindent {\rm (v)} $\mathbb{V}(Z_{i,j}) =\sigma^2 + 3\sigma_u^2 \|\bm s_i \|^2$.   
\end{lem}

\begin{proof}
    Part (i) is clear from the definition of $\mathcal{D}_{\mathcal{R},\sigma_u^2I_n, -\bm c/2}$. For (ii), by Lemma \ref{sum_prod_dist}, we have $\mathbb{E}(2 U_j + C_j) = 2 \mathbb{E}(U_j) + \mathbb{E}(C_j)  = - \mathbb{E}(C_j) + \mathbb{E}(C_j) = 0$ for $0 \le j \le n-1$.

    
\noindent (iii) $\mathbb{V}(2U_j+C_j)=\mathbb{E}((2U_j+C_j)^2)-(\mathbb{E}(2U_j+C_j))^2=4\mathbb{E}(U_j^2)+4\mathbb{E}(U_jC_j)+\mathbb{E}(C_j^2)$ as $\mathbb{E}(2 U_j + C_j) = 0$ by part (ii).  It is clear that $\mathbb{E}(C_j^2)=\frac12$. 
We now compute $\mathbb{E}(U_jC_j)$ and $\mathbb{E}(U_j^2)$ as follows.
    Since $\bm u \leftarrow \mathcal{D}_{\mathcal{R},\sigma_u^2 I_n,-\bm c/2}$, then
    $ u_jc_j=\left \{   \begin{array}{ll}
         u_j & \mbox{if $c_j=1$} \\ 0 & \mbox{if $c_j=0$}  \end{array}. \right. $  
    Therefore, $\mathbb{E}(U_jC_j)=\frac12 (-\frac12)=-\frac{1}{4}$. 
    Note that $u_j \leftarrow \mathcal{N}(-\frac12, \sigma_u^2)$ if $c_j=1$ (happens with probability $\frac12$), and $u_j \leftarrow \mathcal{N}(0, \sigma_u^2)$ if $c_j=0$. If $u_j \leftarrow \mathcal{N}(-\frac12, \sigma_u^2)$, then $\mathbb{E}(U_j)=-\frac12$ and $\mathbb{E}(U_j^2) = \sigma_u^2 +  (\mathbb{E}(U_j))^2 = \sigma_u^2 + \frac14$. Similarly, if $u_j \leftarrow \mathcal{N}(0, \sigma_u^2)$, then $\mathbb{E}(U_j^2) = \sigma_u^2 +  (\mathbb{E}(U_j))^2 = \sigma_u^2$. Overall, we have $\mathbb{E}(U_j^2)= \frac12(\sigma_u^2 + \frac14 + \sigma_u^2) = \sigma_u^2 + \frac18$.
    Hence, $\mathbb{V}(2U_j+C_j)=4\mathbb{E}(U_j^2)+4\mathbb{E}(U_jC_j)+\mathbb{E}(C_j^2)=4\cdot (\sigma_u^2+\frac{1}{8})+4\cdot (-\frac{1}{4})+\frac12= 4\sigma_u^2$.

    
\noindent (iv) 
Let $0 \le j \le n-1$.  
Recall that $\mathbf{y}$ is sampled from $\mathcal{D}_{\mathcal{R}^k, \mathbf{\Sigma}(\mathbf{s},\sigma, \sigma_u) , \bm 0}$ and $\bm s_0 = \bm 1$. 
By Lemma \ref{LemMultivariate} (ii),  $Y_{0,j}$ is the normal distribution $\mathcal{N}(0, \sigma^2-\sigma_u^2 \|\bm s_0\|^2) = \mathcal{N}(0, \sigma^2-\sigma_u^2)$.
Thus, $\mathbb{V}(Z_{0,j})=\mathbb{V}(Y_{0,j}+(2U_j+C_j))=\mathbb{V}(Y_{0,j})+\mathbb{V}(2U_j+C_j) = \sigma^2 - \sigma_u^2 + 4 \sigma_u^2 = \sigma^2 + 3 \sigma_u^2$.

    
\noindent (v) 
As $\mathbf{y}$ is sampled from $\mathcal{D}_{\mathcal{R}^k, \mathbf{\Sigma}(\mathbf{s},\sigma, \sigma_u) , \bm 0}$, then by Lemma \ref{LemMultivariate} (ii),
$Y_{i,j}$ is the normal distribution $\mathcal{N}(0, \sigma^2-\sigma_u^2 \| \bm s_i\|^2)$ for $1 \le i \le k-1$ and $0 \le j \le n-1$. We have $\mathbb{V}(Y_{i,j})=\sigma^2-\sigma_u^2 \|\bm s_i \|^2$. Letting $w_j=2u_j+c_j$, we have $z_{i,j}=y_{i,j}+\sum_{l+m=j \bmod n} \varepsilon_{l,m} w_ls_{i,m}$. 
Hence 
\begin{align*}
   \mathbb{V}(Z_{i,j}) &= \mathbb{V}\left(Y_{i,j}+\sum_{l+m=j \bmod n} \varepsilon_{l,m} W_ls_{i,m}\right) \\
   &= \mathbb{V}(Y_{i,j})+\mathbb{V}(\sum_{l+m=j \bmod n} \varepsilon_{l,m} W_ls_{i,m}) \\
   &= \sigma^2-\sigma_u^2 \|\bm s_i \|^2+ \sum_{l+m=j \bmod n} s_{i,m}^2 \mathbb{V}(W_l)  \\
   &=  \sigma^2-\sigma_u^2\|\bm s_i \|^2 + 4\sigma_u^2 \|\bm s_i \|^2 \\
   &= \sigma^2 + 3\sigma_u^2 \|\bm s_i \|^2.
\end {align*}
\end{proof}


\begin{thm} \label{exp_ratio}
   Let $0 \le j \le n-1$, \ $1 \le i \le k-1$, and $w_j=2u_j+c_j$. Then $\mathbb{E}(\frac{Z_{i,j}}{Z_{0,0}})=s_{i,j}\cdot \frac{\sigma^2_{W_0}-\sigma^2_u}{\sigma_{Z_{0,0}}^2}$, where $\sigma_{W_0}$ and $\sigma_{Z_{0,0}}$ are the standard deviation of $W_0$ and $Z_{0,0}$ respectively. 
\end{thm}

\begin{proof}
    Recall that $Z_{i,j}$ is a normal distribution. Moreover, for $i,j \ne 0$, we observe that $Z_{i,j}$ and $Z_{0,0}$ are correlated as they have a common $w_0$ term.
    Let $\mu_{Z_{i,j}}, \mu_{Z_{0,0}}$ and $\sigma_{Z_{i,j}}, \sigma_{Z_{0,0}}$ be the mean and standard deviation of $Z_{i,j}, Z_{0,0}$ respectively.
    By Proposition \ref{ratio_normal}, the ratio $\frac{Z_{i,j}}{Z_{0,0}}$ of $Z_{i,j}$ and $Z_{0,0}$ is the Cauchy distribution $C(\alpha_{i,j}, \beta_{i,j})$, where 
    $\alpha_{i,j}= \rho_{i,j} \frac{\sigma_{Z_{i,j}}}{\sigma_{Z_{0,0}}}$ ,  $\beta_{i,j}=\frac{\sigma_{Z_{i,j}}}{\sigma_{Z_{0,0}}} \sqrt{1-\rho_{i,j}^2}$ , and  $\rho_{i,j}=\frac{\mathbb{E}((Z_{i,j}-\mu_{Z_{i,j}})(Z_{0,0}-\mu_{Z_{0,0}}))}{\sigma_{Z_{i,j}} \sigma_{Z_{0,0}}}$.
    By considering $\frac{Z_{i,j}}{Z_{0,0}}$ as a truncated Cauchy distribution, we have $\mathbb{E}(\frac{Z_{i,j}}{Z_{0,0}} ) = \alpha_{i,j}$.

    We shall now derive a formula for $\alpha_{i,j}$. Let $\bm s_i=(s_{i,0}, \ldots, s_{i,n-1})$ and $w_j=2u_j+c_j$. Recall from Equation \eqref{EqZ} that $z_{0,0}= y_{0,0}+w_0$ and $z_{i,j}=y_{i,j}+w_0s_{i,j}+\sum_{l+m=j, l \ne 0} \varepsilon_{l,m} w_ls_{i,m}$. We first compute $\rho_{i,j}$ as follows
    \begin{eqnarray*}
  \rho_{i,j}  &=& \frac{\mathbb{E}((Z_{i,j}-\mu_{Z_{i,j}})(Z_{0,0}-\mu_{Z_{0,0}}))}{\sigma_{Z_{i,j}} \sigma_{Z_{0,0}}} = \frac{\mathbb{E}(Z_{i,j}Z_{0,0})}{\sigma_{Z_{i,j}} \sigma_{Z_{0,0}}} \\
   &=& \frac{1}{\sigma_{Z_{i,j}} \sigma_{Z_{0,0}}}\mathbb{E} {\Big (} {\Big (} Y_{i,j}+W_0s_{i,j} \\
   & & +\sum_{l+m=j \bmod n, l \ne 0} \varepsilon_{l,m} W_ls_{i,m} {\Big )} (Y_{0,0}+W_0) {\Big )} \\
   &=& \frac{\mathbb{E}(Y_{i,j}Y_{0,0}) + \mathbb{E}(W_0^2s_{i,j})}{\sigma_{Z_{i,j}} \sigma_{Z_{0,0}}}=\frac{-s_{i,j}\sigma^2_u+s_{i,j}\mathbb{E}(W_0^2)}{\sigma_{Z_{i,j}} \sigma_{Z_{0,0}}} \\
   &=& \frac{s_{i,j}(\sigma_{W_0}^2-\sigma^2_u)}{\sigma_{Z_{i,j}} \sigma_{Z_{0,0}}}.
\end{eqnarray*}
    In the above, we used the fact that $W_0$ and $Y_{i,j}$ (for $0 \le i \le k-1$, $0 \le j \le n-1$) are independent, so that $\mathbb{E}(W_0 Y_{i,j}) = 0$; and similarly $W_0$ and $W_{l}$ (for $1 \le l \le n-1$) are independent, so that $\mathbb{E}(W_0W_{l}) = 0$ for $l \ne 0$.

    We can now compute $\mathbb{E}(\frac{Z_{i,j}}{Z_{0,0}})=\alpha_{i,j}$ as follows: 
\begin{eqnarray*}        
    \mathbb{E}\left(\frac{Z_{i,j}}{Z_{0,0}}\right) &=& \alpha_{i,j} \\
    &=& \rho_{i,j} \frac{\sigma_{Z_{i,j}}}{\sigma_{Z_{0,0}}} 
    = \frac{s_{i,j}(\sigma_{W_0}^2-\sigma^2_u)}{\sigma_{Z_{i,j}}  \sigma_{Z_{0,0}}} \frac{\sigma_{Z_{i,j}}}{\sigma_{Z_{0,0}}} \\
   &=& s_{i,j} \cdot \frac{\sigma_{W_0}^2-\sigma^2_u}{\sigma_{Z_{0,0}}^2}.
\end{eqnarray*}
    
\end{proof}

By Theorem \ref{exp_ratio}, $\mathbb{E}(\frac{Z_{i,j}}{Z_{0,0}})=s_{i,j} \cdot \frac{\sigma_{W_0}^2-\sigma^2_u}{\sigma_{Z_{0,0}}^2}$ for $1 \le i \le k-1, \, 0 \le j \le n-1$. To determine the required number of signatures for the ratio attack, let $L > 0$ and consider the truncated Cauchy distribution $\mathcal{C}_{i,j}=\mathcal{C}(\alpha_{i,j}, \beta_{i,j} \, | \, [\alpha_{i,j}-L, \alpha_{i,j}+L])$ with $\alpha_{i,j}=s_{i,j} \alpha_*$ and $\beta_{i,j}=\frac{\sigma_{Z_{i,j}}}{\sigma_{Z_{0,0}}}\sqrt{1-(\frac{s_{i,j} \, (\sigma_{W_0}^2-\sigma^2_u)}{\sigma_{Z_{0,0}} \, \sigma_{Z_{i,j}}})^2}$, where $\alpha_*:=\frac{\sigma_{W_0}^2-\sigma^2_u}{\sigma_{Z_{0,0}}^2}$ and by Lemma \ref{LemEV} (v), $\sigma_{Z_{i,j}}=\sqrt{\sigma^2+3\sigma_u^2 \|\bm s_i \|^2} \approx \sqrt{\sigma^2+\sigma_u^2\eta(\eta+1)n}$. Let $\bar{U}_1,\ldots, \bar{U}_{N}$ be identically distributed truncated Cauchy distribution $\mathcal{C}_{i,j}$ and $\hat{U}=\frac{1}{N}\sum_{\hat{k}=1}^N \bar{U}_{\hat{k}}$ with mean $\alpha_{i,j}$ and standard deviation $\sigma_{i,j}$. 
By Proposition \ref{truncated_cauchy_to_normal}, $\sqrt{N}(\hat{U}-\alpha_{i,j})$ is approximate to the normal distribution $\mathcal{N}(0, \sigma_{i,j}^2)$, where $\sigma_{i,j,L}= \sqrt{\frac{\beta_{i,j} L}{\arctan(\frac{L}{\beta_{i,j}})}-\beta_{i,j}^2}$.  

Let $\sigma_z$ be the standard deviation of $Z_{i,j}$. Take $L=\lceil  l \sigma_z \rceil$ for a real number $l >2$, then by Lemma \ref{L12}, we have $\Pr(|z| > L) \le \Pr(|z| > l \sigma_z) < e^{-l^2/2}$.
    Note that the attack is successful if $|\hat{U}-\alpha_{i,j}| < \frac{\alpha_*}{2}$. To achieve this, according to Theorem \ref{ThmNSample}, we need $N_{i,j,L}=(\frac{\omega \sigma_{i,j,L}}{\alpha_*/2})^2$ samples.
    To ensure that all the $N_{i,j,L}$ samples lie within the truncated interval $[\alpha_{i,j}-L, \alpha_{i,j}+L]$, we impose the condition $N_{i,j,L} \cdot e^{-l^2/2} < 1$.
    We thus take $L_*=\min_l \{ L=\lceil  l \sigma_z \rceil \, | \, N_{i,j,L} \cdot e^{-l^2/2} < 1 \}$ and compute $p_{i,j}=\displaystyle\frac{1}{\sigma_{i,j,L_*}\sqrt{2\pi}} \int_{-\alpha_* /2}^{\alpha_* /2} e^{-\frac{1}{2}(\frac{t}{\sigma_{i,j,L_*}})^2} {\rm d}t$.  Let $p_*=\min_{i \ne 0, j} \{p_{i,j} \}$. Then $O(1/p_{*}^2)$ will give an approximation of the required number of signatures. Thus, we have the following theorem.

\begin{thm} \label{ThmN}
   Let $l$ be a real number such that $l > 2$ and $\sigma_z$ be the standard deviation of $Z_{i,j}$ for $1 \le i \le k-1, 0 \le j \le n-1$. Let $L=\lceil  l \sigma_z \rceil$ and $L_*=\min_l \{ L=\lceil  l \sigma_z \rceil \, | \, N_{i,j,L} \cdot e^{-l^2/2} < 1 \}$, where $N_{i,j,L}$ is defined above. Then $O(1/p_{*}^2)$ gives an approximation for the required number of signatures for the key recovery attack, where $p_*=\min_{i\ne 0,j} \{p_{i,j}\}$ and $p_{i,j}$ is as defined above. 
\end{thm}

\section{Recovering Secret Key of G+G Convoluted Gaussian Signature}\label{SectAttack}

In this section, we describe our technique to recover the secret key of the G+G convoluted Gaussian signature.


By Theorem \ref{exp_ratio}, we have $\mathbb{E}(\frac{Z_{i,j}}{Z_{0,0}} ) = s_{i,j} \alpha_*$, 
where $\alpha_* =\frac{\sigma^2_{W_0}-\sigma^2_u}{\sigma_{Z_{0,0}}^2}$ and $W_0=2U_0+C_0$. We may compute $\alpha_*$ using Lemma \ref{LemEV} (iii) and (iv). Moreover, by this lemma, we observe that $\alpha_* = \frac{\sigma_{W_0}^2-\sigma^2_u}{\sigma_{Z_{0,0}}^2} = \frac{3 \sigma_u^2}{\sigma^2 + 3 \sigma_u^2} <1$ as $\sigma > \sigma_u$.

Since $\alpha_*$ can be computed, the equation $\mathbb{E}(\frac{Z_{i,j}}{Z_{0,0}} ) = s_{i,j} \alpha_*$ allows us to recover the secret key $s_{i,j}$ by collecting a number $N$ of signatures $\mathbf{z}^{(1)}, \ldots, \mathbf{z}^{(N)}$ such that $z_{0,0}^{(\hat{k})} \ne 0$ for $1 \le \hat{k} \le N$  and computing $\mathbb{E}(\frac{Z_{i,j}}{Z_{0,0}}) \approx \frac{1}{N}\sum_{\hat{k}=1}^N \frac{z_{i,j}^{(\hat{k})}}{z_{0,0}^{(\hat{k})}}$.
We can then decide that the secret key value is $s_{i,j}$ if $\frac{1}{N}\sum_{\hat{k}=1}^N \frac{z_{i,j}^{(\hat{k})}}{z_{0,0}^{(\hat{k})}} \in [s_{i,j}\alpha_*-\frac{\alpha_*}{2}, s_{i,j}\alpha_*+\frac{\alpha_*}{2})$. In this way, the secret key of the G+G convoluted Gaussian can be recovered.


As mentioned in Section \ref{G+G}, there is no concrete parameter given for $\mathbf{z}= \mathbf{y} + (2\bm u + \bm  c)\mathbf{s}$ in \cite{DPS23}. Hence, in order to provide a proof-of-concept of our ratio attack,  we carried out some experiments on our ratio attack on the G+G convoluted Gaussian signature for different parameters.  As sampling from multivariate discrete Gaussian distributions in SageMath is time-consuming, it takes a long time to generate a sufficient number of signatures for simulating the attack. We thus simulate the ratio attack on some scaled-down parameters. 
Note that this does not affect the validity of our ratio attack.


In our SageMath \cite{SMath} simulation, we take $k=3, m=1$ and the other parameters are given in Table \ref{TbParamSimul}.


\begin{table}[h]
    \centering
    \caption{The parameters and the number $N$ of signatures used in the proof-of-concept of the ratio attack} \label{TbParamSimul}
    \begin{tabular}{|c|c|c|c|} \hline
  $n$        & $64$     &  $128$   &  $128$  \\ \hline
  $\eta$     & $1$      &  $1$     &  $1$    \\ \hline
  $\sigma_u$ &  $1.0$   &  $0.4$   &  $1.0$   \\ \hline
  $\sigma$   &  $15.0$  & $9.5$   & $24.0$  \\ \hline
  $\sigma_{w_0}$ &  $2.0$    &  $0.8$   &  $4.0$   \\ \hline
  $\sigma_{z_{0,0}}$ &  $15.09$    & $9.525$    &  $24.06$  \\ \hline
  $\alpha_*=\frac{\sigma^2_{w_0}-\sigma^2_u}{\sigma_{z_{0,0}}^2}$ &  $0.0131578$  &  $0.0052904$   &  $0.0051813$  \\ \hline
  $1/p_*^2$ &  $3.15$ mil   & $11.83$ mil   &  $32.12$ mil  \\ \hline
  \# signatures  &  $9.34$ mil     &  $30.66$ mil   &  $125.88$ mil \\ \hline
  Attack Time &   $3.54$ hrs   &   $19.17$ hrs  &  $94.53$ hrs  \\ \hline
    \end{tabular}    
\end{table}

In Table \ref{TbParamSimul}, the number of signatures $1/p_*^2$ is obtained from Theorem \ref{ThmN} with $\omega=3.8905$. The "\# signatures" is the simulated number of signatures required to recover the secret key. The attack time is not optimized and includes the time for generating the signatures.

The important parameters for the ratio attack are $\alpha_*, \sigma_u, \sigma, \eta$ as they determine the required number of signatures. To prevent the ratio attack, we need to set $\alpha_*$ to be small and  the required number of signatures to be large, say $\ge 2^{64}$.  We give one parameter by setting $1/p_*^2 \approx 2^{64}$ for $n=256, \eta=1$ and choosing $\sigma_u=15.0, \sigma=5.3\times10^4$. Then $\alpha_*=2.40298 \times10^{-7}$ and $q$ is obtained from Lemma \ref{L12} by setting $t=9.5$ such that $q >5.03 \times10^5$.
This indicates that the parameters must be chosen to be large, resulting in large key and signature sizes in order to avoid the ratio attack.
It is noted that this parameter is not meant for setting security level, it is just to prevent the ratio attack which needs to collect at least $2^{64}$ signatures.  




\section{Comments on updated Module G+G Signature in the e-print \cite{DPS23e}} \label{Sect_com_DPS23e}

To realize the generic G+G convoluted Gaussian signature, a concrete example of the module-LWE is given in the updated version in the eprint \cite{DPS23e}. Its signature is $\mathbf{z} = \mathbf{y} + (\bm \zeta \bm u + \bm c) \mathbf{s}$, where $\bm \zeta=1+x^{n/2} \in \mathcal{R}$. 
The revised eprint \cite{DPS23e} introduces four major changes compared to the Asiacrypt 2023 paper \cite{DPS23}:

\medskip

\noindent ($\mathcal{C}1$) $\bm u \leftarrow \mathcal{D}_{\mathcal{R}, \sigma_u I_n, - \bm \zeta^* \bm c /2 }$ where $\bm \zeta^*=1-x^{n/2}$. 

\medskip

\noindent ($\mathcal{C}2$) The covariance matrix $\mathbf{\Sigma}(\mathbf{s},\sigma,\sigma_u):=\sigma^2 I_{nk} - \sigma_u^2 {\rm circ}(\bm \zeta \mathbf{s}){\rm circ}(\bm \zeta \mathbf{s})^T$, where $\mathbf{s}=(\bm{s}_1,\cdots, \bm{s}_k)$, where 


\[ {\rm circ}(\bm \zeta \mathbf{s}){\rm circ}(\bm \zeta \mathbf{s})^T:= 2\begin{bmatrix} I_n & \bm{s}_1^T & \ldots & \bm{s}_{k-1}^T \\ \bm{s}_1 & \bm{s}_1\bm{s}_1^T & \ldots & \bm{s}_1\bm{s}_{k-1}^T \\ \vdots & \vdots & \ddots & \vdots \\ \bm{s}_{k-1} & \bm{s}_{k-1}\bm{s}_1^T & \ldots & \bm{s}_{k-1}\bm{s}_{k-1}^T \end{bmatrix}. \]

\medskip

\noindent ($\mathcal{C}3$) The public key is $\mathbf{A}=((2(\mathbf{a}-\mathbf{b}_1)+q\mathbf{J})^T\, | \, 2\mathbf{A}_0 \, | \, 2\mathbf{I}_m )$ and the secret key is $\mathbf{s}=(1 \, | \, \mathbf{s}_1 \, | \, \mathbf{s}_2-\mathbf{b}_0)$, where $\mathbf{J}=(\bm{\zeta}^*,0,\cdots, 0) \in \mathcal{R}^m$, $\mathbf{b}^T=\mathbf{a}^T+\mathbf{A}_0\mathbf{s}_1^T+\mathbf{s}_2^T \bmod q$, $\mathbf{a} \in \mathcal{R}^m$ and $\mathbf{b}$ is decomposed to $(\mathbf{b}_0,\mathbf{b}_1)$ depending on $d$ as follows.

\medskip

\noindent (a) $d=0$: $\mathbf{b}_0=\bm{0}$ and $\mathbf{b}_1=\mathbf{b}$.

\medskip

\noindent (b) $d=1$: Let $\mathbf{b}=(b_0, \cdots, b_{nm-1})$, $\mathbf{b}_i=(b_{i0}, \cdots, b_{i(nm-1)})$ for $i=0,1$.  For $0 \le j \le nm-1$,

\noindent (b1) If $b_j\equiv0 \bmod 2$, then $b_{0j}=0$ and $b_{1j}=b_j$,

\medskip

\noindent (b2) If $b_j \equiv1 \bmod 2$, then 
\[b_{0j} = \left \{   \begin{array}{ll}
        b_j-1  & \mbox{if \, $b_j\equiv1 \bmod 4$} \\ b_j+1  & \mbox{if \, $b_j\equiv3 \bmod 4$}  \end{array}, \right. \]
\[b_{1j} = \left \{   \begin{array}{ll}
        1  & \mbox{if \, $b_j\equiv1 \bmod 4$} \\ -1  & \mbox{if \, $b_j\equiv3 \bmod 4$}  \end{array}. \right. \]

\medskip

\noindent ($\mathcal{C}4$) Parameters:  The updated parameters are given as follows.

\begin{table} [h]
    \centering
    \caption{The parameters of the Module G+G signatures scheme} \label{TbParam_ex}
    \begin{tabular}{|c|c|c|c|} \hline
  Sec        & $120$     &  $180$   &  $256$  \\ \hline
  $n$        & $256$     &  $256$   &  $256$  \\ \hline
  $q$        & $64513$   &  $50177$ &  $202753$  \\ \hline
  $S$        & $82.74$   &  $90.65$ &  $79.75$  \\ \hline
  Keygen Acceptance Rate        & $0.5$   &  $0.5$ &  $0.5$  \\ \hline
  $\sigma_u$        & $14.22$   &  $14.22$ &  $14.22$   \\ \hline
  $\sigma$        & $664.18$   &  $727.68$ &  $640.14$   \\ \hline
  $\gamma$     & $31972.19$      &  $39405.92$     &  $38437.36$    \\ \hline 
  $(m,k-m)$     & $(3,4)$      &  $(4,5)$     &  $(4,7)$    \\ \hline   
  $\eta$     & $1$      &  $1$     &  $1$    \\ \hline 
  $d$     & $1$      &  $1$     &  $0$    \\ \hline    
    \end{tabular}    
\end{table}

\medskip

We would like to comment on the revised module G+G signature scheme concerning the four changes above.  We will show that the signature is secure against the ratio attack. But the parameters given are not practical due to the following reasons: (1) the probability of the covariance matrix being positive definite is very small, (2) the generated signature is not a valid signature and vulnerable to forgery attack. We will discuss these in detail as follows.

\medskip

\noindent {\bf (A) The Module G+G Signature Is Secure Against The Ratio Attack}

\medskip
It is noted that $\bm c$ is sampled from $\{\bm c= (c_0, \ldots, c_{\frac{n}{2}-1}, 0, \ldots, 0) \in \mathcal{R}\footnote{For $\bm a = \sum_{j=0}^{n-1}a_j x^j \in \mathcal{R}$, we sometimes write it as $\bm a = (a_0, a_1, \ldots, a_{n-1})$ for simplicity.} \mid c_j \in \{0,1\} \text{ for } j=0, \ldots, \frac{n}{2}-1\}$.
Thus, $\mathbb{E}(C_j)=\frac{1}{2}$ for $j=0, \ldots, \frac{n}{2}-1$ and $\mathbb{E}(C_j)=0$ for $j=\frac{n}{2}, \ldots, n-1$. 

\begin{lem} \label{polymult}
   Let $\bm u=(u_0,\ldots, u_{n-1})$, then 

\noindent {\rm (i)}  $\bm \zeta^* \bm u=(u_0+u_{\frac{n}{2}}, \ldots, u_{\frac{n}{2}-1}+u_{n-1}, \ -u_0+u_{\frac{n}{2}-1}, \ldots, -u_{\frac{n}{2}-1}+u_{n-1})$.

\noindent {\rm (ii)}  $\bm \zeta \bm u=(u_0-u_{\frac{n}{2}}, \ldots, u_{\frac{n}{2}-1}-u_{n-1}, \ u_0+u_{\frac{n}{2}-1}, \ldots, u_{\frac{n}{2}-1}+u_{n-1})$. 
\end{lem}

\begin{proof}
  Both (i) and (ii) are obtained by straightforward polynomial multiplications in $\mathcal{R}$.
\end{proof}

\begin{cor} \label{zeta*_c}
   Let $\bm c=(c_0,\ldots, c_{\frac{n}{2}-1},0, \ldots, 0)$, then $-\bm \zeta^* \bm c=(-c_0, \ldots, -c_{\frac{n}{2}-1}, c_0, \ldots, c_{\frac{n}{2}-1})$.
\end{cor}

Let $\psi_j=(\bm \zeta \bm u)_j$ be the $i$-th coordinate of $\bm \zeta \bm u$ and denote by $\Psi_j$ the distribution of $\psi_j$, for $j=0,\ldots, n-1$.

\begin{lem}\label{LemEV_ex}
    Suppose $\bm c \leftarrow \mathcal{U}(\mathcal{C})$,  $\bm u \leftarrow \mathcal{D}_{\mathcal{R},\sigma_u^2 I_n,-\bm \zeta^* \bm c/2}$, $\mathbf{y} \leftarrow  \mathcal{D}_{\mathcal{R}^k, \mathbf{\Sigma}(\mathbf{s},\sigma, \sigma_u) , \bm 0}$ and $\| \mathbf{s} \|_{\infty} \le \eta$ . Then, for $0 \le j \le n-1$ and $1 \le i \le k-1$,
   
\noindent {\rm (i)} $\mathbb{E}(U_j) = \left \{   \begin{array}{ll}
        -\frac12 \mathbb{E}(C_j)  & \mbox{if \, $0 \le j \le \frac{n}{2}-1$} \\ \frac12 \mathbb{E}(C_{j-\frac{n}{2}}) & \mbox{if \, $\frac{n}{2} \le j \le n-1$}  \end{array}. \right. $
       
\noindent {\rm (ii)} $\mathbb{E}(\Psi_j + C_j) = 0$,
       
\noindent {\rm (iii)} $\mathbb{V}(\Psi_j + C_j) = 2 \sigma_u^2$,

\noindent {\rm (iv)} $\mathbb{V}(Z_{0,j}) = \sigma^2$, 
       
\noindent {\rm (v)} $\mathbb{V}(Z_{i,j}) = \sigma^2.$ 
\end{lem}

\begin{proof}

\noindent {\rm (i)} It is clear from Corollary \ref{zeta*_c} and the definition of $\mathcal{D}_{\mathcal{R},\sigma_u^2 I_n,-\bm \zeta^* \bm c/2}$.

\noindent {\rm (ii)} By Lemma \ref{polymult} (ii), $$\psi_j = \left \{   \begin{array}{ll}  u_j-u_{\frac{n}{2}+j}  & \mbox{if \, $0 \le j \le \frac{n}{2}-1$} \\ u_{j-\frac{n}{2}}+u_{j} & \mbox{if \, $\frac{n}{2} \le j \le n-1$}  \end{array}. \right.$$ By part (i), we have 

$\mathbb{E}(\Psi_j) = \left \{   \begin{array}{ll}
        \mathbb{E}(U_j-U_{\frac{n}{2}+j})=-\mathbb{E}(C_j)  & \mbox{if \, $0 \le j \le \frac{n}{2}-1$} \\ \mathbb{E}(U_{j-\frac{n}{2}}+U_{j})=0 & \mbox{if \, $\frac{n}{2} \le j \le n-1$}  \end{array}. \right.$  Thus, $\mathbb{E}(\Psi_j + C_j) = \mathbb{E}(\Psi_j) + \mathbb{E}( C_j)=0$.
 
\medskip
        
\noindent (iii) Since $\mathbb{E}((\Psi_j+C_j)^2)=\mathbb{E}(\Psi_j^2)+\mathbb{E}(C_j^2)+2\mathbb{E}(\Psi_jC_j)$,  
\begin{equation}
\psi_j = \left \{   \begin{array}{ll}  u_j-u_{\frac{n}{2}+j}  & \mbox{if \, $0 \le j \le \frac{n}{2}-1$} \\ u_{j-\frac{n}{2}}+u_{j} & \mbox{if \, $\frac{n}{2} \le j \le n-1$}  \end{array} \right.  \label{eqn2}
\end{equation}
 and $\bm u$ is sampled from  the distribution $\mathcal{D}_{\mathcal{R},\sigma_u^2 I_n,-\bm \zeta^* \bm c/2}$ with variance $\sigma_u^2$,  
then, $\mathbb{E}(\Psi_j^2)=\left \{  \begin{array}{ll} \mathbb{E}(U_j^2)-2\mathbb{E}(U_jU_{\frac{n}{2}+j})+\mathbb{E}(U_{\frac{n}{2}+j}^2)  & \mbox{if \, $0 \le j \le \frac{n}{2}-1$} \\ \mathbb{E}(U_{j-\frac{n}{2}}^2)+2\mathbb{E}(U_jU_{j-\frac{n}{2}})+\mathbb{E}(U_j^2)  & \mbox{if \, $\frac{n}{2} \le j \le n-1$}  \end{array}. \right.$ We consider two cases for $c_j$ as follows.

\noindent (a) $c_j=0$: By Corollary \ref{zeta*_c} and Lemma \ref{LemMultivariate} (ii), $u_j \leftarrow \mathcal{N}(0,\sigma_u^2)$. Thus, $\mathbb{E}(U_j)=0$, $\mathbb{E}(U_j^2)=\sigma_u^2+(\mathbb{E}(U_j))^2=\sigma_u^2$, $\mathbb{E}(\Psi_j C_j)=0$ for $0 \le j \le n-1$ and 
\begin{eqnarray*}
 \mathbb{E}(U_jU_{\frac{n}{2}-1})=0 \;\; \mbox{if \, $0 \le j \le \frac{n}{2}-1$,} \\
 \mathbb{E}(U_jU_{j-\frac{n}{2}})=0 \; \; \mbox{if \, $\frac{n}{2} \le j \le n-1$.} 
\end{eqnarray*}
Therefore, $\mathbb{E}(\Psi_j^2)=2\sigma_u^2$.

\noindent (b) $c_j=1$: By Corollary \ref{zeta*_c} and Lemma \ref{LemMultivariate} (ii), 

$u_j \leftarrow \left \{  \begin{array}{ll} \mathcal{N}(-\frac{1}{2},\sigma_u^2)  & \mbox{if \, $0 \le j \le \frac{n}{2}-1$} \\ \mathcal{N}(\frac{1}{2},\sigma_u^2) & \mbox{if \, $\frac{n}{2} \le j \le n-1$}  \end{array} \right.$. Thus, 

$\mathbb{E}(U_j)= \left \{  \begin{array}{ll} - \frac{1}{2}  & \mbox{if \, $0 \le j \le \frac{n}{2}-1$} \\ \frac{1}{2} & \mbox{if \, $\frac{n}{2} \le j \le n-1$}  \end{array}, \right.$
\begin{eqnarray*}
 \mathbb{E}(U_jU_{\frac{n}{2}-1})=(-\frac{1}{2})(\frac{1}{2})=-\frac{1}{4} \;\; \mbox{if \, $0 \le j \le \frac{n}{2}-1$,} \\
 \mathbb{E}(U_jU_{j-\frac{n}{2}})=(\frac{1}{2})(-\frac{1}{2})=-\frac{1}{4} \; \; \mbox{if \, $\frac{n}{2} \le j \le n-1$,} 
\end{eqnarray*}
$\mathbb{E}(U_j^2)= \left \{   \begin{array}{ll}  \sigma_u^2+(\mathbb{E}(U_j))^2=\sigma_u^2+\frac{1}{4}  & \mbox{if \, $0 \le j \le \frac{n}{2}-1$} \\ \sigma_u^2+(\mathbb{E}(U_j))^2=\sigma_u^2+\frac{1}{4} & \mbox{if \, $\frac{n}{2} \le j \le n-1$}  \end{array}. \right.$
Then, 
\begin{align*}
 \mathbb{E}(\Psi_j^2)
&=\left \{  \begin{array}{ll} \sigma_u^2+\frac{1}{4}-2(-\frac{1}{4})+\sigma_u^2+\frac{1}{4}  & \mbox{if \, $0 \le j \le \frac{n}{2}-1$} \\ \sigma_u^2+\frac{1}{4}+2(-\frac{1}{4})+\sigma_u^2+\frac{1}{4}  & \mbox{if \, $\frac{n}{2} \le j \le n-1$}  \end{array} \right. \\
&=\left \{  \begin{array}{ll} 2\sigma_u^2+1  & \mbox{if \, $0 \le j \le \frac{n}{2}-1$} \\ 2\sigma_u^2 & \mbox{if \, $\frac{n}{2} \le j \le n-1$}  \end{array}. \right. 
\end{align*}
By equation \eqref{eqn2},

$\mathbb{E}(\Psi_jC_j)= \left \{  \begin{array}{ll} - \frac{1}{2}- \frac{1}{2}=-1  & \mbox{if \, $0 \le j \le \frac{n}{2}-1$} \\ \frac{1}{2}- \frac{1}{2}=0 & \mbox{if \, $\frac{n}{2} \le j \le n-1$}  \end{array}. \right.$

\medskip
\medskip

Combining (a) and (b), we have 
\begin{align*}\mathbb{E}(\Psi_j^2)& =\begin{cases} \frac{1}{2}(2\sigma_u^2+2\sigma_u^2+1)=2\sigma_u^2+\frac{1}{2} & \mbox{if  $0 \le j \le \frac{n}{2}-1$} \\ 
\frac{1}{2}(2\sigma_u^2+2\sigma_u^2)=2\sigma_u^2  & \mbox{if  $\frac{n}{2} \le j \le n-1$}  \end{cases},\end{align*}

$\mathbb{E}(\Psi_jC_j)= \left \{  \begin{array}{ll}  \frac{1}{2}(-1)=-\frac{1}{2}  & \mbox{if \, $0 \le j \le \frac{n}{2}-1$} \\ 0 & \mbox{if \, $\frac{n}{2} \le j \le n-1$}  \end{array}. \right.$

Hence, 
\begin{align*}\mathbb{V}(\Psi_j+C_j) 
&= \mathbb{E}((\Psi_j+C_j)^2)-(\mathbb{E}(\Psi_j+C_j))^2 \\
&= \mathbb{E}(\Psi_j^2)+\mathbb{E}(C_j^2)+2\mathbb{E}(\Psi_jC_j) \\
&= \left \{  \begin{array}{ll} 2\sigma_u^2+\frac{1}{2}+\frac{1}{2}+ 2(-\frac{1}{2}) & \mbox{if $0 \le j \le \frac{n}{2}-1$} \\ 2\sigma_u^2+0+0  & \mbox{if  $\frac{n}{2} \le j \le n-1$}  \end{array} \right. \\
&= \left \{  \begin{array}{ll} 2\sigma_u^2& \mbox{if $0 \le j \le \frac{n}{2}-1$} \\ 2\sigma_u^2  & \mbox{if  $\frac{n}{2} \le j \le n-1$}  \end{array}. \right.
\end{align*}

\noindent (iv) 
Let $0 \le j \le n-1$. Recall that $\mathbf{y} \leftarrow \mathcal{D}_{\mathcal{R}^k, \mathbf{\Sigma}(\mathbf{s},\sigma, \sigma_u) , \bm 0}$ and $\bm s_0 = \bm 1$. By Lemma \ref{LemMultivariate} (ii),  $Y_{0,j}$ is the normal distribution $\mathcal{N}(0, \sigma^2-2\sigma_u^2 \|\bm s_0\|^2) = \mathcal{N}(0, \sigma^2-2\sigma_u^2)$. Thus, $\mathbb{V}(Z_{0,j})=\mathbb{V}(Y_{0,j}+(\Psi_j+C_j))=\mathbb{V}(Y_{0,j})+\mathbb{V}(\Psi_j+C_j) = \sigma^2 - 2\sigma_u^2 + 2 \sigma_u^2 = \sigma^2$.

\noindent (v) For $1 \le i \le k-1$, $\mathbf{y} \leftarrow \mathcal{D}_{\mathcal{R}^k, \mathbf{\Sigma}(\mathbf{s},\sigma, \sigma_u) , \bm 0}$, then by Lemma \ref{LemMultivariate} (ii), $Y_{i,j}$ is the normal distribution $\mathcal{N}(0, \sigma^2-2\sigma_u^2 \| \bm s_i\|^2)$ for $0 \le j \le n-1$. Thus, $\mathbb{V}(Y_{i,j})=\sigma^2-2\sigma_u^2 \| \bm s_i\|^2$. Letting $w_j=\psi_j+c_j$, we have $z_{i,j}=y_{i,j}+\sum_{l+m=j \bmod n} \varepsilon_{l,m} w_ls_{i,m}$. 

Hence 
\begin{align*}
   \mathbb{V}(Z_{i,j}) &= \mathbb{V}\left(Y_{i,j}+\sum_{l+m=j \bmod n} \varepsilon_{l,m} W_ls_{i,m}\right) \\
   &= \mathbb{V}(Y_{i,j})+\mathbb{V}(\sum_{l+m=j \bmod n} \varepsilon_{l,m} W_ls_{i,m}) \\
   &= \sigma^2-2\sigma_u^2 \|\bm s_i \|^2+ \sum_{l+m=j \bmod n} s_{i,m}^2 \mathbb{V}(W_l)  \\
   &=  \sigma^2-2\sigma_u^2\|\bm s_i \|^2 + 2\sigma_u^2 \|\bm s_i \|^2 \\
   &= \sigma^2. 
\end{align*}
\end{proof}

We now compute $\mathbb{E}(Z_{i,j}/Z_{0,0})$ in this setting by following the proof of Theorem \ref{exp_ratio}. Let $W_j:= \Psi_j + C_j$ for $0 \le j \le n-1$.
In the proof of Theorem \ref{exp_ratio}, we have $\mathbb{E}(Z_{i,j}Z_{0,0})$ is equal to
\[\mathbb{E} {\Big (} {\Big (} Y_{i,j}+W_0s_{i,j} +\sum_{l+m=j \bmod n, l \ne 0} \varepsilon_{l,m} W_ls_{i,m} {\Big )} (Y_{0,0}+W_0) {\Big )}\]
in the computation of $\rho_{i,j}$. 
Then, we have
\[\mathbb{E}(Z_{i,j}Z_{0,0})=\mathbb{E}(Y_{i,j}Y_{0,0})+\mathbb{E}(W_0^2)s_{i,j}.\]
From the covariance matrix $\mathbf{\Sigma}(\mathbf{s},\sigma,\sigma_u)$, we have $\mathbb{E}(Y_{i,j}Y_{0,0})=-2\sigma_u^2 s_{i,j}$. By Lemma \ref{LemEV_ex} (iii), $\mathbb{E}(W_0^2)=2\sigma_u^2$. Therefore, $\mathbb{E}(Z_{i,j}Z_{0,0})=0$.
Following the proof of Theorem \ref{exp_ratio}, consequently, we have $\rho_{i,j}=0$ and so $\mathbb{E}(Z_{i,j}/Z_{0,0}) = 0$. So, we cannot deduce any information about the secret key $s_{i,j}$ from $\mathbb{E}(Z_{i,j}/Z_{0,0})$. Hence, the ratio attack cannot be applied in this case.

\medskip

\medskip

\noindent {\bf (B) Almost Impossible To Generate A Signature ($d=0$)}

\medskip

We performed a simulation, where we generate 50,000 $\bm{\zeta} \mathbf{s}$ for $d=0$ ($260$-bit security) and found that all of the covariance matrices $\mathbf{\Sigma}(\mathbf{s},\sigma, \sigma_u)$ are not positive definite. This indicates that the chance of generating $\mathbf{y}$ for producing a signature is extremely slim ($\ll 1/50000$) for the parameter given in the fourth column of Table \ref{TbParam_ex}.
We also found that the mean and standard deviation of $\sigma_1(\bm{\zeta} \mathbf{s})$ are $82.88$ and $4.20$ respectively.  From Table \ref{TbParam_ex}, $\frac{\sigma}{\sigma_u}=45.016$, which is less than $\sigma_1(\bm{\zeta} \mathbf{s})$ for all 50,000 values obtained in our simulation.    This also shows that $\sigma_1(\bm{\zeta} \mathbf{s}) < S=79.75$ cannot guarantee that the covariance matrix $\mathbf{\Sigma}(\mathbf{s},\sigma, \sigma_u)$ is positive definite.

\medskip

\medskip

\noindent {\bf (C) Generated Signatures Are Invalid ($d=1$) } 

\medskip

In $\mathcal{C}3$, for $d=1$, $\| \mathbf{s}_2-\mathbf{b}_0 \|_{\infty}$ is not small and is in fact close to $\frac{q-1}{2}$. Let $\bar{q}:=\frac{q-1}{2}$.  Then, the 2-norm of the signature $\mathbf{z}=\mathbf{y}+(\bm{\zeta u}+\bm{c})\mathbf{s}$ is 
$$\| \mathbf{z} \|_2 \ge \| \mathbf{z}_2 \|_2 \ge \| \mathbf{s}_2-\mathbf{b}_0 \|_2 \approx \sqrt{m\frac{n}{2}\frac{\bar{q}(\bar{q}+1)}{3}},$$ 
where $\mathbf{z}_2=\mathbf{y}_2+(\bm{\zeta u}+\bm{c})(\mathbf{s}_2-\mathbf{b}_0)$.  A signature
$\mathbf{z}$ is a valid signature if $\| \mathbf{z} \|_2 \le \gamma$. But from the Table \ref{TbParam_ex}, $\| \mathbf{z} \|_2 > 364940.63 > 31972.19 =\gamma$ and $\| \mathbf{z} \|_2 > 327754.79 > 39405.92 =\gamma$ for 120-bit and 180-bit security respectively. Hence, the signature is always invalid.

\medskip

Moreover, for $\| \mathbf{s} \|_{\infty}$ close to $\frac{q-1}{2}$, then the module G+G signature is  vulnerable to forgery attack.
From the public key, we get $2(\mathbf{a}-\mathbf{b}_1)+q\mathbf{J} \bmod 2q$. As $q\mathbf{J}$ is publicly known, we can obtain $2(\mathbf{a}-\mathbf{b}_1) \bmod 2q$ and so $(\mathbf{a}-\mathbf{b}_1) \bmod q$. Let $\bar{\mathbf{a}}:=(\mathbf{a}-\mathbf{b}_1) \bmod q$.
We may then forge a signature as follows.

Randomly choose $\bar{\mathbf{s}}_1 \leftarrow \chi_{\eta}^{k-m-1}$ and compute $\bar{\mathbf{s}}_2=\bar{\mathbf{a}}-\mathbf{A}_0\bar{\mathbf{s}}_1$. Thus, $\bar{\mathbf{s}}_2 \in [-\frac{q-1}{2}, \frac{q-1}{2} ]^{nm}$.  Then, generate a signature by following the original signature generation (refer to ${\sf Sign}(\mathbf{A},\mathbf{s}, \mu)$ in \cite{DPS23e}), using $\bar{\mathbf{s}}_1, \bar{\mathbf{s}}_2$ in place of $\mathbf{s}_1, \mathbf{s}_2$.  Thus, we may easily forge a signature in this case.

\medskip

\medskip

In view of the findings (A), (B), (C) above, we conclude that the parameters of the module G+G signature scheme does not work. 
Fixing it would result in much larger public key and signature sizes than the ones given in \cite{DPS23e}.

\section{Conclusion}\label{SectConclusion}

In this paper, we proposed ratio attack on the generic G+G convoluted Gaussian signature \cite{DPS23}.  We exploit the correlation among the signatures and take the ratio of these signatures which follows a Cauchy distribution. This enables us to find the formula relating the expected value of the ratio of two signatures and the secret key via the truncated Cauchy distribution.  We also proved a formula for computing the required number of signatures to successfully recover the secret key.  

In order to provide a proof-of-concept of the ratio attack on the G+G convoluted Gaussian signature, we implemented the ratio attack in SageMath on some scaled-down parameters of the G+G convoluted Gaussian signature. This demonstrates that the secret key can be completely recovered.  This also shows that the G+G convoluted Gaussian signature \cite{DPS23} is insecure if the parameters are not chosen correctly. Although the revised signature in the revised eprint paper \cite{DPS23e} is secure against the ratio attack, we found that either a valid signature cannot be produced or a signature can be forged easily for their given parameters in \cite{DPS23e}. 



\vspace{0.5cm}

\end{document}